\documentclass[journal]{IEEEtran}
\newtheorem{definition}{Definition}
\newtheorem{theorem}{Theorem}
\newenvironment{proof}{\begin{IEEEproof}}{\end{IEEEproof}}
\usepackage{amssymb}
\newtheorem{example}{Example}  
\newtheorem{strategy}{Strategy}    
\newtheorem{lemma}{Lemma}  
\usepackage{multirow}  
\usepackage{hyperref}

\usepackage{amsmath}
\usepackage{algorithm}
\usepackage{algorithmic}

\ifCLASSINFOpdf
\usepackage[pdftex]{graphicx}
\DeclareGraphicsExtensions{.pdf,.jpeg,.png}
\else
\usepackage[dvips]{graphicx}
\DeclareGraphicsExtensions{.eps}
\fi

\ifCLASSINFOpdf
\else
\fi

\begin{document}

\newtheorem{property}{Property}
%

\title{Utility-driven Data Analytics on Uncertain Data} 
%
%
        
\author{Wensheng Gan,
        Jerry Chun-Wei~Lin,~\IEEEmembership{Member,~IEEE,}        
        Han-Chieh Chao,~\IEEEmembership{Senior Member,~IEEE}\\
        Athanasios V. Vasilakos,~\IEEEmembership{Senior Member,~IEEE}
        and Philip S. Yu,~\IEEEmembership{Fellow,~IEEE}

\thanks{Wensheng Gan is with the Department of Computer Science, University of Illinois at Chicago, IL, USA. (E-mail: wsgan001@gmail.com)}

\thanks{Jerry Chun-Wei Lin is with the Department of Computing, Mathematics, and Physics, Western Norway University of Applied Sciences, Bergen, Norway. (E-mail: jerrylin@ieee.org)}

\thanks{Han-Chieh Chao is with the Department of Electrical Engineering, National Dong Hwa University, Hualien, Taiwan. (E-mail: hcc@ndhu.edu.tw)}

\thanks{Athanasios V. Vasilakos is with the Department of Computer Science, Electrical and Space Engineering, Lulea University of Technology, Sweden. (E-mail: th.vasilakos@gmail.com)}

\thanks{Philip S. Yu is with the Department of Computer Science, University of Illinois at Chicago, IL, USA. (E-mail: psyu@uic.edu)}

\thanks{Manuscript received XXXX; revised XXXX.}
}

\markboth{IEEE Internet of Things Journal, 2019}%
{Shell \MakeLowercase{\textit{et al.}}: Bare Demo of IEEEtran.cls for IEEE Journals}

\maketitle

\begin{abstract}

Modern Internet of Things (IoT) applications generate massive amounts of data, much of it in the form of objects/items of readings, events, and log entries. Specifically, most of the objects in these IoT data contain rich embedded information (e.g., frequency and uncertainty) and different levels of importance (e.g., unit utility of items, interestingness, cost, risk, or weight). Many  existing approaches in data mining and analytics have limitations such as only the binary attribute is considered within a transaction, as well as all the objects/items having equal weights or importance. To solve these drawbacks, a novel utility-driven data analytics algorithm named HUPNU is presented, to extract \textbf{\underline{H}}igh-\textbf{\underline{U}}tility patterns by considering both \textbf{\underline{P}}ositive and \textbf{\underline{N}}egative unit utilities from \textbf{\underline{U}}ncertain data. The qualified high-utility patterns can be effectively discovered for risk prediction, manufacturing management, decision-making, among others. By using the developed vertical Probability-Utility list with the Positive-and-Negative utilities (PU$^{\pm}$-list) structure, as well as several effective pruning strategies. Experiments showed that the developed HUPNU approach performed great in mining the qualified patterns efficiently and effectively.

\end{abstract}


%
\IEEEpeerreviewmaketitle

\section{Introduction}
%
%
%

\IEEEPARstart{W}{ith} the increasing prevalence of sensors, mobile phones, actuators, and RFID tags, Internet of Things (IoT) applications generate massive amounts of rich data per day \cite{atzori2010internet,tsai2014data,chen2015data}. Examples include the control signals issued to Internet-connected devices like lights or thermostats, measurements from industrial and medical equipment, manufacturing data, or log files from smart-phones that record the complex behavior of sensor-based applications. These applications need to find frequent patterns that represent typical system behavior (e.g., \cite{tsai2014data,agrawal1994fast,han2004mining}) as well as high risk/utility patterns that represent very important knowledge from normal behavior \cite{ahmed2009efficient,liu2012mining,tseng2013efficient}. For industrial areas \cite{chen2015data,sheng2013survey}, various algorithms have been designed by applying data mining methods \cite{chen2015data,fournier2017survey0}. These methods can be used to evaluate the complicated and complex data from the industry and solve existing problems.

In some industrial areas, manufacturing schedules can be planned by the discovered information and knowledge, which can then be used to increase and gain utilities, maximally \cite{chen2015data,yun2017efficient}. Pattern mining and analysis, from an industrial research perspective, provides a unique and up-to-date insight into the manufacturing industry \cite{chen2015data}. There are many industrial utility-driven manufacturing systems with different IoT applications. Consider an environment with all types of sensors to monitor abnormal conditions. Let each transaction denote the set of sensors showing above normal value at a particular instance, and the value associated with each sensor can be the abnormality or risk measure. Here the number of units can mean the number of abnormal sensors of the same type, then how to discover the high risk patterns is quite useful for manufacturing analytics in  Internet of Things. In the retail industry, it can identify the purchase behaviors of customers, which can be used to make specific decisions and improve the service quality of customers. However, existing Frequent-based Pattern Mining (FPM) algorithms \cite{fournier2017survey0} ignore that information, and the discovered information or pattern may contain useless knowledge, such as items with a lower utility that may be discovered.  Thus, traditional FPM cannot handle the problem of the quantitative databases, and also fails to extract the utility-driven patterns which are insightful for data analytics. A new utility-driven data analytics framework named High-Utility Pattern Mining (HUPM) \cite{ahmed2009efficient,liu2012mining,tseng2013efficient} was presented to provide a solution to the above limitations. It is important to notice that the \textit{utility} concept can be referenced to reflect multiple participants' satisfaction \cite{tao2012utility}, utility \cite{yun2017efficient,2gan2018survey}, revenue \cite{teng2018revenue}, risk \cite{2gan2018survey}, or weight \cite{cai1998mining}.

\subsection{Motivation}
To reveal high utility/risk patterns for decision-making, it considers the quantitative database, as well as the unit utility of the items. The analytics of using HUPM instead of FPM is quite helpful for system monitoring, planning, manufacturing management, and decision-making \cite{tsai2014data,chen2015data,yun2017efficient}. Most algorithms in traditional HUPM consider all items to have positive unit utilities/risks. When the items have the constraint of negative values, for example sales discount (e.g., buy two get one free) or if a supermarket/retail store sells the products at a loss to stimulate the sales of other relative items (e.g., sell printers to promote the notebook/PC), traditional HUPM cannot mine meaningful information. This situation happens in real-life scenarios, especially when promoting some profitable items for gaining more money in cross-selling conditions. Some prior works have mentioned that traditional HUPM algorithms \cite{ahmed2009efficient,tseng2013efficient} may generate incomplete and missing information \cite{chu2009efficient} when the dataset consists of negative unit utility objects/items.

For analyzing the complex industrial data and manufacturing data \cite{tsai2014data,yun2017efficient,wan2017manufacturing}, it may encounter various challenges, i.e., the embedded uncertainty, positive risk and negative risk, and other factors. The uncertainty factor exists in many realistic situations, such as the collection of noisy data sources (e.g., GPS, wireless sensor network, RFID, or WiFi systems) \cite{aggarwal2009survey,bernecker2009probabilistic}. Traditional pattern mining algorithms cannot be utilized straightforwardly to inaccurate or uncertain environments for mining the required knowledge or information. The reason for this is that the \textit{utility} can be considered as a semantic measure to value how the ``utility'' of a pattern based on users' priori experience, goal, and knowledge, and the \textit{uncertainty} can be considered as an objective measure to value the probability of a pattern as an objective existence; they are two totally different factors. Most existing utility-based algorithms have been studied to handle precise data, while they are unable to deal with data uncertainty \cite{lin2016efficient}. If the uncertain factor was not considered in the mining process, it may find useless or misleading information with low existential probability.

\subsection{Contribution}

To solve the above limitations and problems, we attempt to design a novel algorithm named HUPNU to extract \underline{\textbf{H}}igh-\underline{\textbf{U}}tility patterns with both \underline{\textbf{P}}ositive and \underline{\textbf{N}}egative unit utilities from \underline{\textbf{U}}ncertain data in intelligent environment. Moreover, to deal with realistic situations in real-life applications, the positive and negative unit utilities are considered. The significant contributions of this work are listed below: 

\begin{itemize}
	 \item To the best of our knowledge on pattern mining, this is the first study to discuss the problem by considering both the positive and negative unit utilities, and the uncertainty factor to discover the qualified high-utility patterns. This approach is more suitable for realistic situations such as risk prediction, manufacturing management, e-commerce, and decision making.
	 \item A vertical and compact \underline{\textbf{P}}robability-\underline{\textbf{U}}tility \underline{\textbf{list}}, with both a \underline{\textbf{p}}ositive and \underline{\textbf{n}}egative utilities (PU$^{\pm}$-list) structure, is developed. This list can keep all the essential knowledge from the database for later mining progress. 
	 \item A one-phase method called HUPNU is designed to discover the qualified high-utility patterns using the PU$^{\pm}$-list structure. The multiple database scans and the generate-and-test mechanism can be largely avoided and ignored. 
	 \item We also developed several pruning strategies to easily remove the unpromising candidates and reduce the search space size of the qualified high-utility patterns. These pruning strategies can also efficiently reduce the size of the PU$^{\pm}$-list by the designed downward closure property. 
	 \item Several experiments were conducted on both synthetic and realistic datasets. Results showed that the developed HUPNU performed great in mining the qualified patterns efficiently and effectively. 
	 
\end{itemize}
 
This paper is organized as follows: A literature review is given in Section \ref{sec:2}. Preliminaries and the problem statement are shown in Section \ref{sec:background}. The designed HUPNU with the several pruning strategies are studied in Section \ref{sec:algorithm}. Extensive experiments are conducted in Section \ref{sec:experiments}. The conclusion is finally provided in Section \ref{sec:conclusion}.

\section{Related Work}
\label{sec:2}
\subsection{Support-based Pattern Mining}
Data mining and analytic technologies are used in many different domains \cite{tsai2014data,yun2017efficient,wan2017manufacturing,liu2017exploring} and they provide powerful ways of discovering useful, meaningful, and implicit information from very large datasets. Frequent Pattern Mining (FPM) is the most fundamental concept in retrieving the qualified information using a support-based constraint, and many works have been developed based on the support criteria to mine frequent itemsets or association rules \cite{agrawal1994fast,han2004mining,fournier2017survey0}.  Other factors, such as interestingness \cite{geng2006interestingness} or weights, have also been considered with mining criteria to find the interesting or important patterns in the task of pattern mining. Many algorithms have been designed to find the meaningful patterns from a binary database \cite{fournier2017survey0,geng2006interestingness,luna2016speeding}. However, they only assess whether an item appears in a transaction, and this approach does not consider the useful factors, for example, an event may be occurred in multiple quantities in a record.  Quantitative Association-Rule Mining (QARM) \cite{srikant1996mining,hong1999mining} was presented instead of the binary value (0 or 1) for discovering more meaningful and useful information.

Different from processing precise data, some pattern mining approaches have been developed to deal with uncertain data, for discovering frequent expected patterns \cite{chui2007mining} or probabilistic  frequent patterns \cite{bernecker2009probabilistic} by taking the uncertainty from data into account. The reason is that the uncertainty factor exists in many realistic data sources (e.g., GPS, wireless sensor network, RFID, or WiFi systems). Some details of uncertain data algorithms and applications can be referred to \cite{aggarwal2009survey}.

\subsection{Utility-based Pattern Mining with Efficiency Issues} 
Although QARM solves the past limitation of traditional association-rule mining, it still does not consider more important and interesting factors such as the unit utilities of the objects/items, which can bring the profitable objects to user in service-oriented
manufacturing system \cite{tao2012utility,2gan2018survey}. In addition, the support-based constraint of pattern mining is inappropriate for measuring the importance of the items in realistic situations. To tackle these problems, High-Utility Pattern Mining (HUPM) \cite{ahmed2009efficient,liu2012mining,tseng2013efficient} was presented to reveal high-utility patterns. HUPM considers both the unit utility and quantity of the objects to show the high-utility patterns from the quantitative database, which can provide more meaningful results than that of the support-based algorithms. Yao et al. \cite{yao2004foundational} first defined the utility-mining problem by considering the occurred quantity (treated as the \textit{internal utility}) and unit utility of the objects/items (treated as the \textit{external utility}) to reveal the itemsets with high utilities. 
In the past decade, HUPM has been considered as the emerging topic in many tasks of data analytics, and some well-known algorithms are developed, such as the  Transaction-Weighted Utility (TWU) model \cite{liu2005two}, IHUP \cite{ahmed2009efficient}, UP-growth and UP-growth+ \cite{tseng2013efficient}, HUI-Miner \cite{liu2012mining}, and so on.
Many variants of HUPM have also been discussed focusing on mining different forms of utility-oriented patterns. The importance of HUPM is increasing, especially in the current era of big data \cite{tsai2015big}, and more opportunities and challenges are required for discussion and analysis since HUPM can provide realistic benefits to the retailers and managers in many different applications and domains \cite{yun2017efficient}.

\subsection{Utility-based Pattern Mining with Effectiveness Issues} 

In addition to the efficiency issue of utility mining, a number of models have been proposed to address the effectiveness issue of discovering different kinds of high-utility patterns (HUPs). Current HUPM  algorithms can successfully handle the temporal data \cite{lin2015efficient,lin2017two}, and dynamic data \cite{lin2015fast,gan2018survey,3lin2016fast}. Other interesting effectiveness issues, such as HUPM with discount strategies \cite{lin2016fast}, the concise representation \cite{tseng2015efficient}, discriminative issue \cite{lin2017fdhup}, and top-$k$ problem \cite{tseng2016efficient} for HUPM, have been extensively studied. Lin et al. first proposed an utility mining model to extract the high-utility patterns from uncertain databases \cite{lin2016efficient}. Different from the above utility measures of HUPs, another utility measure for utility-driven pattern mining namely utility occupancy is introduced recently \cite{gan2018huopm}. And a comprehensive survey of utility mining has been provided by Gan et al. \cite{2gan2018survey}.

All the above HUPM algorithms only consider the positive utilities/risks and quantities of items. However, in some real-world scenarios, the utility/risk/weight values of the objects/items in databases usually can be either positive values or negative values. Therefore, the traditional algorithms of HUPM can not successfully be applied to handle the databases containing negative values. In the past, the two-phase HUINIV-Mine \cite{chu2009efficient}, TS-HOUN \cite{lan2014shelf}, and one-phase FHN \cite{lin2016fhn} algorithms are proposed to deal with precise data which containing both positive and negative utility values. However, all the existing negative-based approaches cannot be used to process uncertain data and extract the utility-driven insightful patterns.


\section{Preliminaries and Problem Statement} 
\label{sec:background}
We first introduce the uncertain database of the defined problem for utility-driven mining in this section.  Let $I$ = $\{i_1, i_2, \dots, i_m\} $ be a set of objects/items (symbols), and let the uncertain database be a set of transactions such as \textit{D} = \{$ T_{1}, T_{2}, \dots, T_{n} $\}, and each object/item in a transaction has an uncertain probability of existence such as $ p(i_{k}, T_{c}) $ \cite{aggarwal2009survey,bernecker2009probabilistic}. For each $ T_{c} $, it has the relationship such that $ i_{k}\in T_{c} $. A positive quantity value is defined as the internal utility, and denoted as $ q(i_{k}, T_{c}) $. This quantity value shows the quantity of $ i_{k} $ in $ T_{c} $. Let $ i_{k}\in I $ be related to a positive or negative value, which is defined as the external utility, and denoted as $ pr(i_{k}) $. A set of external utility of all items in the database is denoted as \textit{ptable} = \{$ pr(i_{1}), pr(i_{2}), \dots, pr(i_{m}) $\}. Table \ref{table:db} shows a simple example to illustrate the proposed approach.

\begin{table}[!htbp]
	\setlength{\abovecaptionskip}{0pt}
	\setlength{\belowcaptionskip}{0pt}	
	\centering
	\small
	\caption{A running example for the uncertain database.} 
	\label{table:db}
	\begin{tabular}{|c|c|c|c|}
		\hline
		\textbf{\textit{tid}} & \textbf{Item: quantity, probability)} & \textbf{total utility}    \\ \hline
		$ T_{1} $ & 	(\textit{a}:5, 0.6); (\textit{b}:3, 0.50); (\textit{d}:2, 0.9); (\textit{e}:4, 0.8)   &  \$107   \\ \hline
		$ T_{2} $ & 	(\textit{c}:1, 0.75); (\textit{d}:1, 0.9); (\textit{e}:2, 1.0)  &  \$24  \\ \hline
		$ T_{3} $ &	    (\textit{a}:4, 1.0); (\textit{b}:3, 1.0); (\textit{c}:2, 0.7);  (\textit{e}:1, 0.75)  &  \$50   \\ \hline
		$ T_{4} $ &	    (\textit{a}:3, 0.9); (\textit{c}:1, 0.9)  &    \$22   \\ \hline
		$ T_{5} $ &	    (\textit{b}:2, 1.0); (\textit{c}:4, 0.95); (\textit{d}:5, 0.6); (\textit{e}:4, 1.0) &   \$90   \\ \hline
	\end{tabular}
\end{table}

\begin{example} 
	Table \ref{table:db} is considered as the running example and can be described as follows: It has five transactions ($ T_1, T_2, \dots, T_5 $). Transaction $T_2$ shows that items $\{c\}$, $\{d\}$, and $\{e\}$ are purchased together in $T_2$, and their quantities are 1, 1, and 2, respectively. We also assume that that unit utilities of the items in the Table \ref{table:db} are defined in \textit{ptable} as: \textit{ptable} = \{$pr$(\textit{a}):\$8, $pr$(\textit{b}):\$5, $pr$(\textit{c}):-\$2, $pr$(\textit{d}):\$12, $pr$(\textit{e}):\$7\}. Thus, it is obvious to see that an item $(c)$ is sold at loss.
\end{example}

\begin{definition}[\textbf{utility measure}]
	The $u(i, T_c)$ indicates the utility of an item $i$ in the transaction $T_c$, which can be calculated as:  $u(i, T_c)$ = $pr(i) \times q(i, T_c)$. $u(X, T_c)$ shows the utility of an itemset $ X $ in a transaction $T_c$, which can be calculated as: $u(X, T_c)$ = $\sum_{i \in X} {u(i, T_c)}$. Note that $X \subseteq I$. Thus, the total utility of $ X $ in a database $ D $ can be denoted as $u(X)$, which can be calculated as: $u(X)$ = $ \sum_{X\subseteq T_{c} \wedge T_{c} \in D} {u(X, T_c)} $.
\end{definition}

\begin{example}
	For example in Table \ref{table:db}, the utility of $\{a\}$ in $T_1$ is calculated as: $u(a, T_1)$ = 5 $ \times $ \$8 = \$40.  The utility of $\{a,e\}$ in $T_1$ is calculated as: $u(\{a,e\}, T_1)$ =  $u(a, T_1)$ + $u(e, T_1)$ = 5 $ \times $ \$8 + 4 $ \times$ \$7 = 6\$8. Therefore, the utility of $\{a,e\}$ in the Table \ref{table:db} can be summed up as: $u(\{a,e\})$ =
	$u(a, T_1)$ + $u(e, T_1)$ + 
	$(u(a, T_3)$ + $u(e, T_3)$ = 
    (\$40 + \$28) + (\$32 + \$7) = \$107. For the $\{a,b,e\}$, the total utility of  $\{a,b,e\}$ can be calculated as: $u(\{a,b,e\})$ =
	$u(a, T_1)$ + $u(b, T_1)$ + $u(e, T_1)$ +
	$u(a, T_3)$ + $u(b, T_3)$ + $u(e, T_3)$ =
	(\$40 + \$15 + \$28) + (\$32 + \$15 + \$7) = \$137.
\end{example}

Since the discovered patterns are usually rare in realistic applications, the probabilistic frequent model \cite{bernecker2009probabilistic} cannot be directly applied for any utility-oriented applications \cite{lin2016efficient}. The common method of mining uncertain data uses the expected support-based model to mine the interesting patterns. For example, the expected support of \textit{X} is to sum up the support value of a pattern \textit{X} in a possible world $W_j$ as: $ expSup(X)$ = $\sum_{i=1} ^{|D|} (\prod _{x_i \in X} p(x_i, T_c)) $  \cite{aggarwal2009survey,bernecker2009probabilistic}. The definition of the expected probability measure of the mentioned problem is defined as below. 

\begin{definition}[\textbf{probability measure}]
	Let \textit{X} be a pattern (itemset) and $T_c$ be a transaction in the database \textit{D}. The probability of \textit{X} in $T_c$ is denoted as: $ p(X, T_c) $, which can be calculated as: $ p(X, T_c)$ = $\prod _{i \in X} p(i, T_c) $. Note that $ X \subseteq I $. The probability of \textit{X} in \textit{D} can thus be denoted as $ Pro(X) $, and defined as $ Pro(X)$ = $ \sum_{T_c \in D} (\prod _{i \in X} p(i, T_c)) $. 
\end{definition}

\begin{example}
	In Table \ref{table:db}, the probability of $\{a\}$ in $T_1$ can be calculated as: $ p(a, T_a)$ = 0.60.  The probability of $(a,e)$ in $T_1$ can then be calculated as: $ p(\{a,e\}, T_1)$ = $ p(a, T_1) \times p(e, T_1)$ = 0.6 $\times $ 0.8 = 0.48. The probability of  $(a)$ in $D$ can be calculated as: $ Pro(a)$ = 2.5, and the probability of $\{a,b,e\}$ in $D$ can be calculated as: $ Pro(\{a,b,e\})$ =  $p(\{a,b,e\}, T_1)$ + $p(\{a,b,e\}, T_3)$ = 0.24 + 0.75 = 0.99. 
\end{example}

\begin{definition}[\textbf{Potential High-Utility Itemset, PHUI}]
   Let \textit{X} be a pattern (itemset) in an uncertain database \textit{D}. We can say that $X$ is a PHUI if it satisfies two conditions: (1) $ u(X) \geq minUtil $, and (2) $ Pro(X) \geq minPro \times |D| $, in which \textit{minUtil} is a minimum utility threshold and \textit{minPro} is a minimum probability threshold. We can then conclude that interesting desired PHUI has a high expected probability and a high utility value.
\end{definition}

\textbf{Problem Statement.} With an uncertain database, a utility table (with a positive or negative utility value of each item), a minimum utility threshold (\textit{minUtiil}), and a minimum probability threshold (\textit{minPro}), the problem of utility-driven data analytics from uncertain data  is to discover the complete set of PHUIs.
 
For example, if the \textit{minPro} and  \textit{minUtil} are respectively set as \textit{minPro} = 0.25 and \textit{minUtil} = 20, the derived PHUIs from Table \ref{table:db} are \{\{$\{a\}$:\$96, 2.50\}; \{$\{b\}$:\$40, 2.50\}; \{$\{d\}$:\$96, 2.40\}; \{$\{e\}$:\$77, 3.55\}; \{$\{a,b\}$:\$102, 1.30\}; \{$\{a,c\}$:\$50, 1.51\}; \{$\{b,e\}$:\$103, 2.15\}; \{$\{c,e\}$:\$35, 2.225\}; \{$\{d,e\}$:\$166, 2.22\}; \{$\{b,c,e\}$:\$48, 1.475\}\}. Here, \{$\{b\}$:\$40, 2.50\} indicates that the utility value and expected probability of pattern $\{b\}$ are \$40 and 2.50, respectively.

\section{Proposed Approach for Mining PHUIs} 
\label{sec:algorithm}
 In this section, an utility-driven data analytics framework named HUPNU is presented to discover the Potential High-Utility Itemsets (PHUIs) from an uncertain database. We further design the Probability-Utility list with positive and negative utilities (PU$^{\pm}$-list). In addition, several pruning strategies are presented here to reduce the search space of the potential HUIs. More details are described below. 

\subsection{Positive and Negative Unit Utilities}
The monotonic/anti-monotonic proprieties cannot be held in utility mining \cite{liu2005two}. In this situation, the utility of a pattern may be higher, lower, or equal to any of its subset patterns. The search space to discover the meaningful and useful patterns may become large if many items exist in the database. The TWU model \cite{liu2005two} was presented to avoid the problem of ``combinational exploration'', which aims at reducing the search space for mining the high-utility itemsets. Several extensions are then extensively studied \cite{liu2012mining,liu2005two,tseng2013efficient} to improve the mining performance. However, those approaches, including the TWU model, do not consider both the positive and negative unit utilities of items, which are addressed in this paper.  Moreover, the existing works do not consider the above situations in the uncertain database for discovering the PHUIs. We define the following properties for the addressed problem as follows:

\begin{property}
We first assume that $pu(X)$ and $nu(X)$ are respectively the sum of positive and negative utility of an item $X$ in a database. Thus, we can obtain that the utility of \textit{X} is calculated as: $u(X)$ = $pu(X) + nu(X)$, where $nu(X) \leq u(X) \leq pu(X)$ holds. 
\end{property}

From the above-stated property, we can conclude that $u(X)$ and $nu(X)$ cannot be straightforwardly used as the over-estimated utility of a pattern (\textit{X}). Moreover, even $ pu(u) $ is the upper-bound value of \textit{X}, the downward closure property for the superset of \textit{X} cannot be held since both the positive and negative unit utilities of the items are considered in the addressed problem. We then re-utilize the traditional TWU \cite{liu2005two} property to establish a new over-estimated value of the discovered pattern.
 
\begin{definition} 
	In HUPM, the transaction utility (\textit{TU}) is defined as: $TU(T_c)$ = $\sum_{i \in T_c}{u(i, T_c)}$. In this paper, by considering both positive and negative unit utility of items, we then re-define the transaction utility as:  $RTU(T_c)$ = $\sum_{i \in T_c \wedge pr(i) >0}{u(i, T_c)}$. The transaction-weight utilization of an itemset \textit{X} is also then redefined as \textit{RTWU}: $RTWU(X)$ = $\sum_{X\subseteq T_{c} \wedge T_{c} \in D}{RTU(T_c)}$. The above-stated definitions can be used to hold the \textit{downward closure} property for mining the required PHUIs. Note that $ RTWU(X) \geq u(X)$. 
\end{definition}

\begin{example}
	For example, the $RTU(T_2)$ is \$26. Consider two patterns $\{a\}$ and $\{a,b,e\}$, the $RTWU(\{a\})$ = \$185 and $RTWU(\{a,b,e\})$ = \$161; both of them are the over-estimated values of  $u(\{a\})$ = \$96 and $u(\{a,b,e\})$ = \$137.
\end{example}

\subsection{Probability Utility (PU$^{\pm}$)-List Structure}

In the developed HUPNU algorithm, the processing order of the items in the database is defined as $\succ$, which holds the properties as: (1) the items are then sorted as the \textit{RTWU}-ascending order; and (2) negative items always succeed positive ones. The designed Probability Utility (PU$^{\pm}$)-List structure used in the HUPNU algorithm is stated as follows:

\begin{definition}[\textbf{PU$^{\pm}$-list}]
	Let \textit{X} be an itemset in the database. The PU$^{\pm}$-list of \textit{X} is denoted as: \textit{X.PUL}, and it consists of five elements: (1) \textit{tid} represents the transaction ID in the database; (2) $pro$ is the expected probability of $X$ in $T_{tid}$, and $pro(X,T_{tid})\geq 0 $; (3) $pu$ shows the positive utility of $X$ in $T_{tid}$, and  $u(X,T_{tid})\geq 0 $; (4) $nu$ represents the negative utility of $X$ in $T_{tid}$, and $u(X,T_{tid}) < 0 $; (5) $rpu$ represents $\sum_{i \in T_{tid} \wedge i \succ x \forall x \in X}{u(i, T_{tid})\geq 0}$, which keeps only a positive utility value for the remaining items. 
\end{definition}

\begin{figure}[!htbp]
	\setlength{\abovecaptionskip}{0pt}
	\setlength{\belowcaptionskip}{0pt}
	\centering
	\includegraphics[scale=0.52]{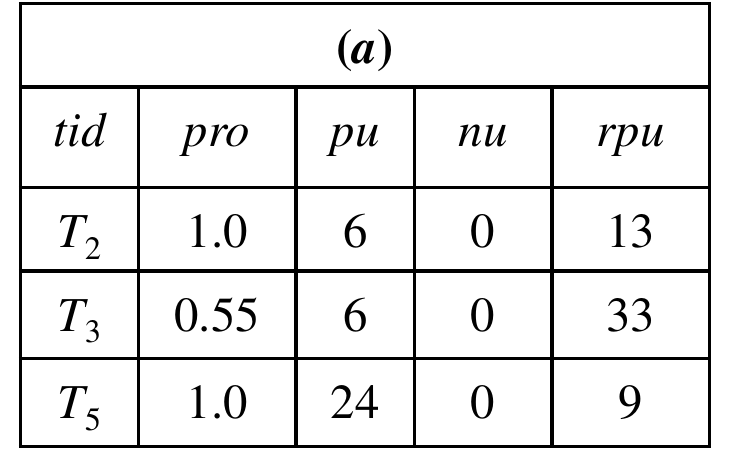}
	\caption{Constructed PU$^{\pm}$-list of pattern $\{a\})$} 
	\label{figPU-list_A}
\end{figure}

\begin{figure}[!htbp]
	\setlength{\abovecaptionskip}{0pt}
	\setlength{\belowcaptionskip}{0pt}
	\centering
	\includegraphics[scale=0.38]{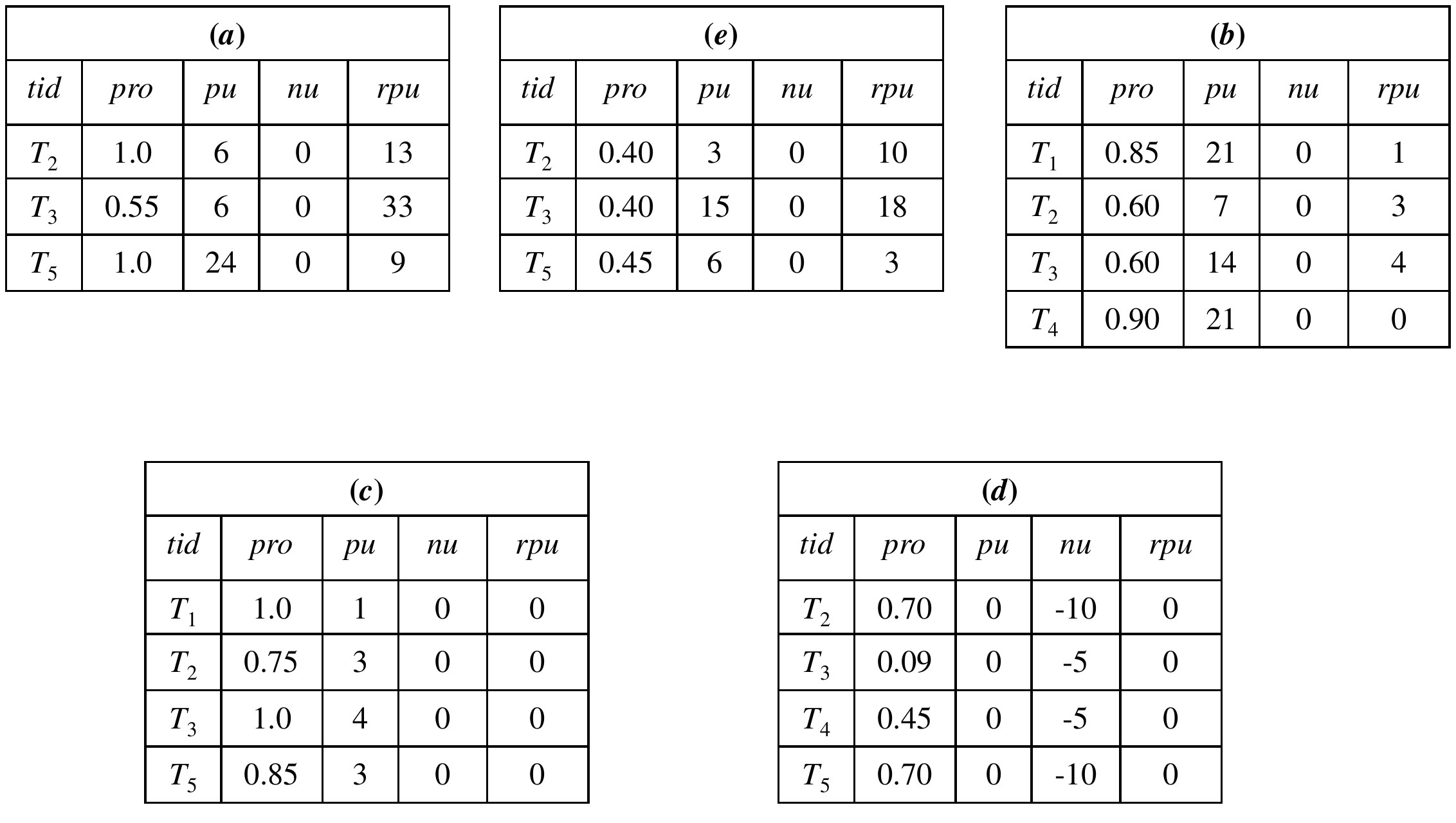}
	\caption{Constructed PU$^{\pm}$-list of the running example} 
	\label{figPU-list_1itemsets}
\end{figure}

\begin{example}
	The search space of the developed HUPNU approach can be shown as the utility-based Set-enumeration tree \cite{lin2017fdhup}, called a PU$^{\pm}$-tree, based on the developed PU$^{\pm}$-list. Since \{$RTWU(a)$: \$185; $RTWU(b)$: \$259; $RTWU(c)$: \$202; $RTWU(d)$: \$231; and $RTWU(e)$: \$285\}, the designed processing order $\succ$ of the running example can be represented as: $\{a \succ d \succ b \succ e \succ c\}$. Thus, the constructed PU$^{\pm}$-list for all items in Table \ref{table:db} is shown in Fig. \ref{figPU-list_1itemsets}. We have $\{a\}.PUL$ = $\{(T_1$, 0.60, \$40, \$0, \$67),
	$(T_3$, 1.00, \$32, -\$4, \$22),
	$(T_5$, 0.90, \$24, -\$2, \$0)\};
	$\{c\}.PUL$ = $\{(T_2$, 0.75, \$0, -\$2, \$0),
	$(T_3$, 0.70, \$0, -\$4, \$0),
	$(T_4$, 0.90, \$0, -\$2, \$0),
	$(T_5$, 0.95, \$0, -\$8, \$0)\};
	$\{a,c\}.PUL$ = $\{T_3$, 0.70, \$32, -\$4, \$0),
	$(T_4$, 0.81, \$24, -\$2, \$0)\}.
\end{example}


\begin{definition} 
	The sums of the total utility, \textit{pu} values, \textit{nu} values, and \textit{rpu} values in the PU$^{\pm}$-list of \textit{X} are respectively denoted as: \textit{SUM(X.iu)}, \textit{SUM(X.pu)}, \textit{SUM(X.nu)}, and \textit{SUM(X.rpu)}, which can be respectively defined as:\\
	$ SUM(X.pu)$ = $\sum_{X\in T_{c}\wedge T_{c}\subseteq D}X.pu(T_{c}) $;\\
    $ SUM(X.nu)$ = $\sum_{X\in T_{c}\wedge T_{c}\subseteq D}X.nu(T_{c}) $;\\
	$ SUM(X.rpu)$ = $\sum_{X\in T_{c}\wedge T_{c}\subseteq D}X.rpu(T_{c}) $;\\
    $ SUM(X.iu)$ = $SUM(X.pu)$ + $SUM(X.nu)$.	
\end{definition}

\begin{lemma}
	\label{lemma-pu}
	For two patterns, such as \textit{Y} and \textit{Z} in the PU$^{\pm}$-tree, if the relationship holds: (1) $ SUM(Y.pu) + SUM(Y.rpu) $ - $ \sum_{\forall T_{c}\in D, Y \subseteq T_{c} \bigwedge Z \nsubseteq T_{c}}(Y.pu + Y.rpu) < minUtil $, or (2) $ SUM(Y.pro) $ - $ \sum_{\forall T_{c}\in D, Y \subseteq T_{c} \bigwedge Z \nsubseteq T_{c}}(Y.pro) < minPro \times |D|$, the (\textit{YZ}), or any superset of it, will not be a PHUI.
\end{lemma}

\begin{strategy}[\textbf{PU-Prune strategy}]
 	If Lemma \ref{lemma-pu} holds, then the PU$^{\pm}$-list construction of the pattern $Y$ can be avoided; and any of its superset will not be a PHUI.
\end{strategy}

According to the designed PU-Prune strategy, the huge size of the unpromising \textit{k}-patterns (\textit{k} $\geq$ 2, and $k$ is the size of the items within the itemset) of the search space can be greatly filtered. Let $Py$ denote an itemset, and $ y $ denote an item; $ Py $ is defined as $ P\cup y $ and \textit{y} is before \textit{z}. The construction procedure of the PU$^{\pm}$-list is shown in Algorithm \ref{AlgorithmConstruct2}. It takes the PU$^{\pm}$-lists of $P$, $Py$, and $Pz$ as the inputs, and returns the PU$^{\pm}$-list of $Pyz$ as output. The PU$^{\pm}$-lists of $k$-itemsets ($k\geq$ 2) can be easily built using a simple join operation of (\textit{k}-1)-itemsets; multiple database scans can be greatly avoided, and the computational cost of the run time can be reduced. If $P.PUL$ is empty, the PU$^{\pm}$-list of a 2-itemset is constructed (Line 9). Otherwise, the PU$^{\pm}$-list of a $k$-itemset ($k \geq $ 3) is constructed (Lines 5 to 7).  For these optimization purposes, the joint operation of the PU$^{\pm}$-lists of $P$, $Py$, and $Pz$ can be constructed by a binary search.


\begin{algorithm}
	\caption{Construction procedure}
	\label{AlgorithmConstruct2}
	\begin{algorithmic}[1]
		\REQUIRE {$ P $, $ Py$, $ Pz $.} 
		\ENSURE {\textit{Pyz} with its \textit{Pyz.PUL}.}
		
		\STATE	$ Pyz.PUL \leftarrow \emptyset$;
		\STATE	set $ Probability $ = \textit{SUM}$(Y.pro)$, $ Utility$ = \textit{SUM}$(Y.pu)$ + \textit{SUM}$(Y.rpu) $;
		
		\FOR {each tuple $ex \in Py.PUL$}
		\IF{$\exists ez \in Pz.PUL$ and $ey.tid$ = $eyz.tid$}
		\IF{$P.PUL \neq \emptyset$ }
		\STATE search each element $e \in P.PUL$ such that $e.tid$ = $ey.tid$.;
		\STATE $ eyz \leftarrow $ $<$$ey.tid, ey.pro \times ez.pro / e.pro$, $ey.pu + ez.pu - e.pu$, $ey.nu  + ez.nu - e.nu$, $ez.rpu$$>$;
		
		\ELSE
		\STATE $ eyz \leftarrow $ $<$$ey.tid, ey.pro \times ez.pro$, $ey.pu + ez.pu$, $ey.nu + ez.nu, ez.rpu$$>$;
		\ENDIF
		\STATE $ Pyz.PUL \leftarrow Pyz.PUL \cup \{eyz\}$;
		
		\ELSE
		\STATE \textit{Probability} = \textit{Probability} - \textit{ey.pro}, \textit{Utility} = \textit{Utility} - \textit{ey.pu} - \textit{ey.rpu};
		\IF{\textit{Probability} $ < minPro  \times |D|$ or $Utility < minUtil $ }
		\STATE \textbf{return} \textit{null}.  
		\ENDIF			
		\ENDIF
		\ENDFOR
		
		\STATE \textbf{return} \textit{Pyz}	
	\end{algorithmic}
\end{algorithm}


\subsection{Proposed Pruning Strategies}

In this section, we discuss several pruning strategies based on the developed PU$^{\pm}$-list for later mining progress. The developed strategies can greatly help to remove the uncompromising candidates in the early stages, and thus the search space that reveals the actual PHUIs can become smaller. The node of the (\textit{k}-1)-itemset in the designed PU$^{\pm}$-tree is denoted as: $X^{k-1} (k \geq 2)$, and any superset of it is denoted as: $ X^{k} $.  

\begin{theorem}[\textbf{Downward closure property of \textit{RTWU} and \textit{probability}}]
	In the PU$^{\pm}$-tree, the correctness of $ Pro(X^{k-1}) \geq Pro(X^{k}) $ and \textit{RTWU}(\textit{X}$ ^{k-1} $) $\geq$ \textit{RTWU}(\textit{X$ ^{k} $}) hold.
 \end{theorem}

\begin{proof}
	Since $ p(X, T_c)$ = $\prod _{i \in X} p(i, T_c) $ can be held for any $T_c$ in $D$, thus we can have that: $ p(X^{k},T_{c}) \leq p(X^{k-1},T_{c}) $. Since $ X^{k-1} $ is the subset of $ X^{k} $; the $tids$ of $ X^k $ is the subset of the $ tids $ of $X^{k-1} $. We then can have that: $Pro(X^{k})$ = $\sum _{X^{k}\subseteq T_{c}\wedge T_{c}\in D}p(X^{k},T_{c})$ $\leq \sum _{X^{k-1}\subseteq T_{c}\wedge T_{c}\in D}p(X^{k-1},T_{c})$ = $Pro(X^{k-1})$. Thus, it can be concluded that \textit{Pro}(\textit{X}$ ^{k-1} $) $\geq$ \textit{Pro}(\textit{X$ ^{k} $}). Moreover, \textit{X}$ ^{\textit{k}-1} $ $\subseteq$ \textit{X$ ^{k} $}, $RTWU(X^{k})$ = $\sum _{X^{k}\subseteq T_{c}\wedge T_{c}\in D}tu(T_{c})$ $\leq \sum _{X^{k-1}\subseteq T_{c}\wedge T_{c}\in D}tu(T_{c})$ = $RTWU(X^{k-1})$ holds.
\end{proof}

\begin{lemma} [\textbf{Upper-bound probability of PHUI}] 
	The summed up probability of any node in the PU$^{\pm}$-tree is greater than the summed up probabilities of its supersets.
\end{lemma}

\begin{strategy}
	The \textit{RTWU} and the probability value of each item can be easily obtained during the initial database scan. Thus, if the summed up probability and \textit{RTWU} of an itemset \textit{X} do not achieve the two conditions of PHUI, \textit{X} and any supersets of it can be pruned directly. 
\end{strategy}

\begin{strategy}
	While the depth-first search is performed to traverse the PU$^{\pm}$-tree, if the summed up probability of tree node \textit{X} such as $ Pro(X)$ is no larger than $ minPro \times |D| $, then none of the supersets of this node are considered to be a PHUI.
\end{strategy}

\begin{lemma} [\textbf{Upper-bound utility of the PHUI}] 
	For a node \textit{X} in the PU$^{\pm}$-tree, the summed values of $SUM(X.pu)$ and $SUM(X.rpu)$ are always equal to or larger than any of its supersets.  
\end{lemma}

We can conclude from the above lemmas that the summed up utility w.r.t. $u(X^{k})$ of an itemset $ X^{k} $ is always no greater than the summed values of \textit{SUM}$(X^{k-1}.pu)$ and \textit{SUM}$(X^{k-1}.rpu)$. Therefore, it is ensured that the transitive extensions with items having positive or negative utilities can hold the downward closure property. The following strategies can be designed based on the above upper-bound constraints:

\begin{strategy}
	While the depth-first search is performed to traverse the PU$^{\pm}$-tree, if the summed up values of \textit{SUM}$(X^{k-1}.pu)$ and \textit{SUM}$(X^{k-1}.rpu)$ are less than $minUtil$, any supersets of \textit{X} cannot be a PHUI. Those nodes in the tree can be directly ignored and pruned to reduce the search space for mining the PHUIs. 
\end{strategy}

\begin{strategy}
	For an itemset \textit{X} in the PU$^{\pm}$-list, if the \textit{X.PUL} is \textit{null} or the $Pro(X)$ value is less than $ minPro \times |D| $, \textit{X} cannot be considered a PHUI; none of its superset (node) is a PHUI. Therefore, the construction procedure of PU$^{\pm}$-lists of \textit{X}'s supersets can be ignored.
\end{strategy}
 
The efficient Estimated Utility Co-occurrence Pruning (EUCP) strategy \cite{fournier2014fhm} is also utilized here for the designed HUPNU algorithm. Thus, the Estimated Utility Co-occurrence Structure (\textit{EUCS}) is built to keep the \textit{RTWU} values of the 2-itemsets. More information can be found in \cite{fournier2014fhm}.

\begin{strategy}
	Let $X$ be a 2-itemset, which is also one of the nodes in the Set-enumeration PU$^{\pm}$-tree. While the depth-first search is performed,  if the \textit{RTWU} of $ X $ is no greater than the $ minUtil $ based on the built EUCS, \textit{X} and any supersets of \textit{X}  are not considered to be the PHUI; the construction progress of the PU$^{\pm}$-tree for \textit{X} and the supersets of it can be ignored.
\end{strategy}

Based on the above proposed pruning strategies, the designed HUPNU is shown in detail below.

\subsection{The Procedure of HUPNU}
In this section, the main procedure of the developed HUPNU is shown in Algorithm \ref{Algorithm_HUPNU}. First, it examines the uncertain database to find the values of \textit{RTWU} (with the redefined $RTU$) and $Pro(i)$  of each item (Line 1). The expected support and \textit{RTWU} of each item in the set $I^*$ are then checked against the $ minPro \times |D| $ and $minUtil$, respectively, and the satisfied items are then discovered and obtained. In this step, the other items can be ignored directly since they could not be the potential HUI (Line 2). The database is then scanned again (Line 4) to re-order the items in the transactions according to the designed order as $\succ$ (Line 3). In addition, the items in the transactions are then re-ordered based on the total order $\succ$ while performing the database scan. After that, the PU$^{\pm}$-list of each 1-item $i \in I^*$ is constructed, respectively, and the depth-first search is recursively performed by the \textit{Search} procedure with the empty itemset $\emptyset$, the set of single items $I^*$, \textit{minPro}, \textit{minUtil}, and the EUCS (Line 5).

\begin{algorithm}
	\caption{HUPNU main procedure}
	\label{Algorithm_HUPNU}
	\begin{algorithmic}[1]
		\REQUIRE {$D$, \textit{minPro}, \textit{minUtil}, \textit{ptable}.} 
		\ENSURE {The set of \textit{PHUIs}.}
		
		\STATE scan uncertain database $D$ to calculate the \textit{RTWU}$(i)$ and $Pro(i)$ of each item $i \in I$;
		\STATE $I^* \leftarrow$ each item $i$ such that $Pro(i) \geq minPro \times |D| \wedge RTWU(i) \geq minUtil $;
		\STATE sort the items in the set of $I^*$ as $\succ$ order; 
		\STATE scan database $D$ to build the PU$^{\pm}$-list of each item $i \in I^*$ and construct $EUCS$;
		\STATE call \textbf{Search($\emptyset$, $I^*$, \textit{minPro}, \textit{minUtil}, \textit{EUCS})}.
		
		\STATE \textbf{return} \textit{PHUIs}	
	\end{algorithmic}
\end{algorithm}

\begin{algorithm}
	\caption{Search procedure}
	\label{AlgorithmSEARCH}
	\begin{algorithmic}[1]
		\REQUIRE {$P$, \textit{ExtensionsOfP}, \textit{minPro}, \textit{minUtil}, \textit{EUCS}.} 
		\ENSURE {The set of \textit{PHUIs}.}
		
		\FOR{each pattern $Py \in $ \textit{ExtensionsOfP}}
		\IF{\textit{SUM(Py.pro)} $ \geq minPro \times |D|$ $ \wedge $ \textit{(SUM(Py.pu)} + \textit{SUM(Py.nu)}) $ \geq minUtil$}
		\STATE output $Py$ as a \textit{PHUI};
		\ENDIF
		\IF{\textit{SUM(Py.pro)} $ \geq minPro \times |D|$ $  \wedge $ (\textit{SUM(Py.pu)} + \textit{SUM(Py.rpu)}) $ \geq minUtil $}
		\STATE \textit{ExtensionsOfPy} $ \leftarrow \emptyset$;
		\FOR{$Pz \in $ \textit{ExtensionsOfP} such that $z \succ y$}
		\IF{$RTWU(\{y,z\}) \geq minUtil$}
		\label{linepruning}  
		\STATE $ Pyz \leftarrow Py \cup Pz $;
		\STATE \textit{Pyz.PUL} $ \leftarrow $  \textbf{Construct(\textit{P}, \textit{Py}, \textit{Pz})};
		\IF{\textit{Pyz.PU}L $ \not= \emptyset $ and \textit{SUM(Pyz.pro)} $ \geq minPro $ $ \times |D| $}
		\STATE \textit{ExtensionsOfPy} $ \leftarrow $ \textit{ExtensionsOfPy} $\cup Pyz $;
		\ENDIF
		\ENDIF	
		
		\ENDFOR
		\STATE \textbf{call Search(\textit{Pyz}, \textit{ExtensionsOfPy}, \textit{minPro}, \textit{minUtil}, \textit{EUCS})}.
		\ENDIF
		\ENDFOR
		\STATE \textbf{return} \textit{PHUIs}	
	\end{algorithmic}
\end{algorithm}

The search procedure of the HUPNU is described in Algorithm \ref{AlgorithmSEARCH}. For each extension $Py$ of $P$, if the probability of $Py$ is no less than $ minPro \times |D| $, and the summed up actual utility of $Py$ in the PU$^{\pm}$-list (denoted as \textit{SUM}$(Y.pu)$ + \textit{SUM}$(Y.nu)$) is no less than $minUtil$, then $Py$ can be considered to be a PHUI (Lines 2 to 4). The developed pruning Strategies 3 and 4 are then performed to check whether the extension ($Py$) can be explored (Line 5). This progress can be executed by integrating $Py$ with all extensions $Pz$ of $P$, such that $z \succ y$ and $RTWU(\{y,z\}) \geq minUtil$ (Line 8, pruning Strategy 6), to generate the extensions ($Pyz$) containing $|Py|+1$ items. After that, the PU$^{\pm}$-list of $Pyz$ can be built by the \textit{Construct} procedure to join the PU$^{\pm}$-lists of $P$, $Py$, and $Pz$ (Lines 9 to 13). Note that the promising PU$^{\pm}$-lists can only be later explored (Line 12, pruning Strategy 5). A recursive $Search$ procedure of $Pyz$ is then performed to obtain its utility and explore the extension(s) (Line 16).

\section{Experimental Study} \label{sec:experiments}

In this section, we evaluate the performance of the developed HUPNU algorithm. All the algorithms were implemented with Java language, and executed on an Intel Core-i5 processor running on Microsoft Windows 7 64 bits operation system with 4GB of main memory. Memory usage was measured by Java API. To the best of our knowledge, this is the first paper that considers both the positive and negative unit utilities of items in an uncertain database; none of the previous works handle this topic. Thus, the developed HUPNU along with several designed pruning strategies were compared in the experiments. HUPNU$ _{P12}$ denotes that pruning Strategies 1 and 2 were involved in the HUPNU; HUPNU$ _{P123}$ considers pruning Strategies 1, 2, and 3; HUPNU$ _{P1234}$ takes the pruning Strategies 1, 2, 3, and 4; and HUPNU$ _{All}$ is concerned with all the pruning Strategies (1 to 5) for evaluation. Experiments were carried out on five realistic datasets\footnote{\url{http://fimi.ua.ac.be/data/}} (i.e., kosarak, accidents, psumb, retail, and mushroom) and one synthetic dataset \cite{IBMdata}. The individual characteristics of each of the six datasets are given below.

\begin{itemize}
	\item \textbf{kosarak:} it has 990,002 transactions, and the number of distinct items is 41,270. The average length and the maximum length of transactions are is 8.09 and 2,498, respectively.
	\item \textbf{accidents:} it has 340,183 transactions and 468 distinct items. For the all transactions, the average length is 33.8 and the maximum length is 51.
	\item \textbf{psumb:} it has total 49,046 transactions and 2,113 distinct items. It is a very dense dataset since the average length of each transaction is 74.
 	\item \textbf{retail:} this dataset contains 88,162 transactions and total 16,470 distinct items. The average length of transactions is 10.3, and the maximum length is 76.
 	\item \textbf{mushroom:} it has 8,124 transactions with 119 distinct items. It is a dense dataset since both the average length and the maximum length are 23. 
  	\item \textbf{T10I4D100K \cite{IBMdata}:} this dataset has 100,000 transactions with 870 items. These transactions have average 10.1 items, and the maximum length is 29. 	
\end{itemize}

The external utilities of different items for the six datasets were generated in the range of [-1,000, 1,000] using a log-normal distribution. In addition, the quantities of the items were randomly generated in the range of [1, 5]. These settings are similar to the previous well-known algorithms \cite{liu2012mining,fournier2014fhm,tseng2013efficient} for HUPM. Moreover, the unique probabilities of the items were randomly assigned in the range of (0.0, 1.0). In the experiments, we evaluated the implemented algorithms in terms of runtime, number of visited nodes (or patterns), and scalability. The results are given below.  

\subsection{Evaluation of Runtime} 
The runtime of the four implemented algorithms were then compared under the variants of  \textit{minUtil} and \textit{minPro} thresholds. Results in terms of the two thresholds are shown in Fig. \ref{fig:Runtime1} and Fig. \ref{fig:Runtime2}, respectively.

\begin{figure*}[!htbp]
	\setlength{\abovecaptionskip}{0pt}
	\setlength{\belowcaptionskip}{0pt}
	\centering
	\includegraphics[trim=90 0 80 12,clip,scale=0.58]{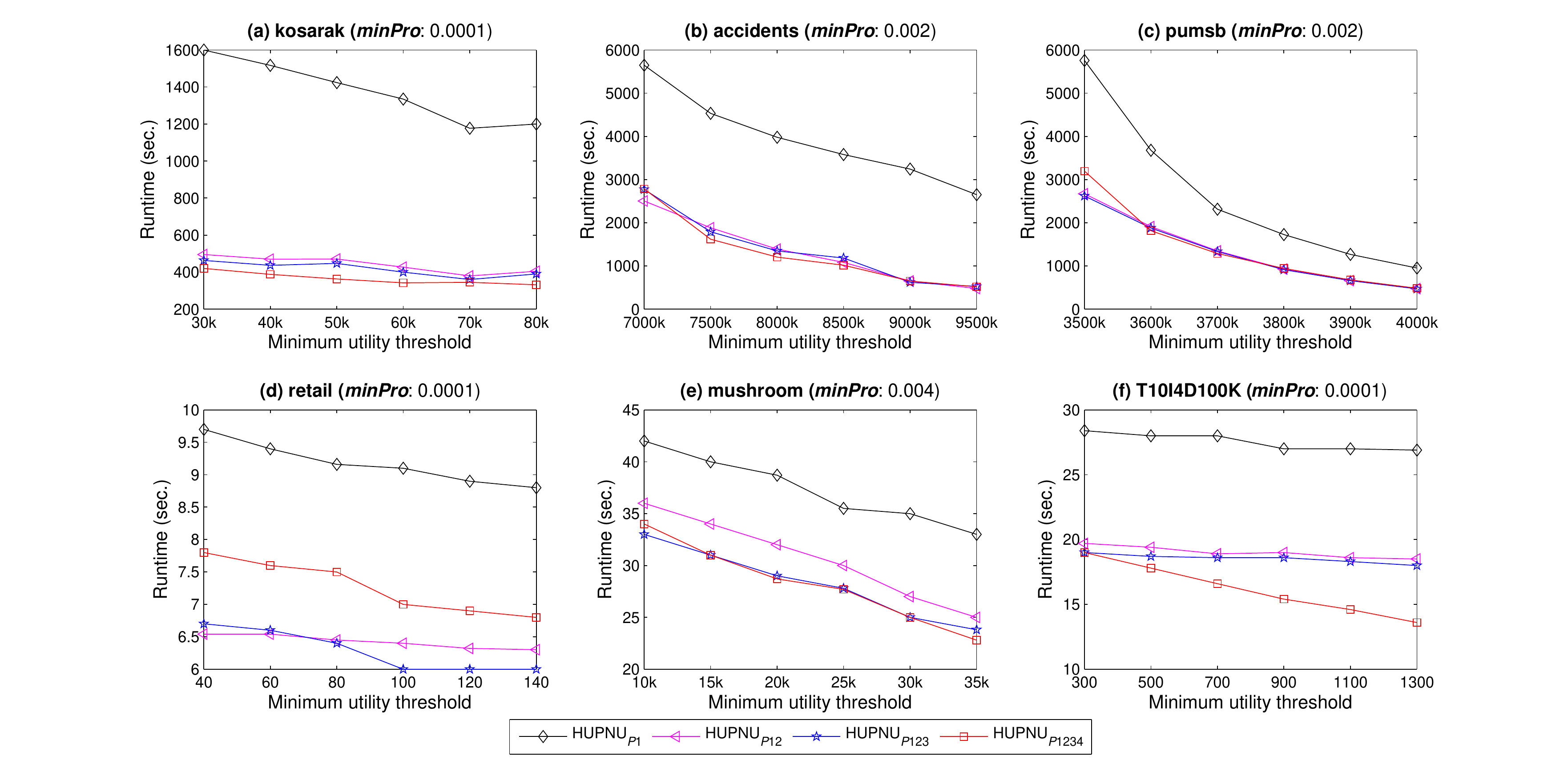}
	\caption{Runtime under varied \textit{minUtil} with a fixed \textit{minPro}.}
	\label{fig:Runtime1}	
\end{figure*}

\begin{figure*}[!htbp]
	\centering
	\setlength{\abovecaptionskip}{0pt}
	\setlength{\belowcaptionskip}{0pt}	
	\includegraphics[trim=90 0 80 12,clip,scale=0.58]{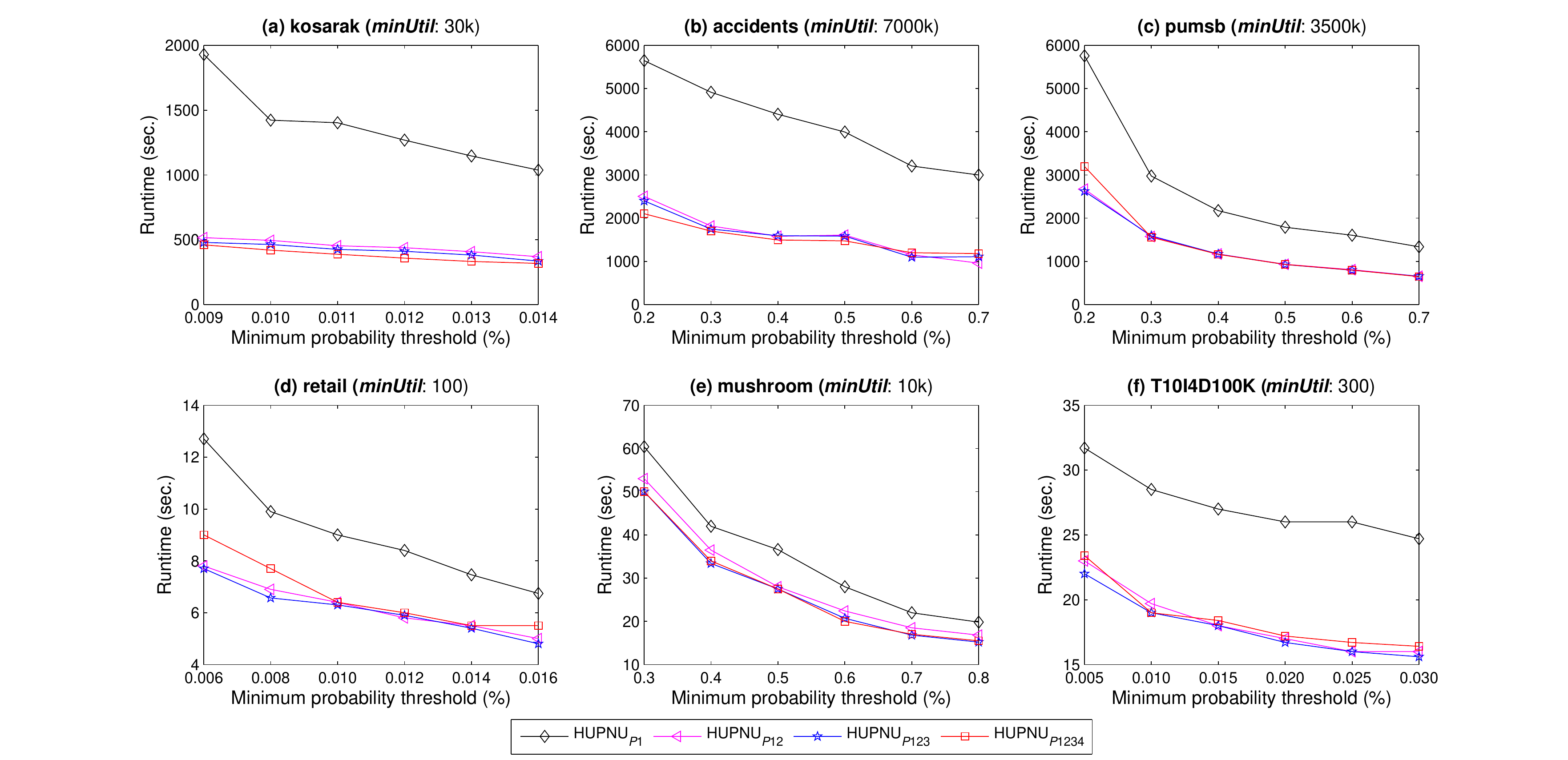}
	\caption{Runtime under varied \textit{minPro} with a fixed \textit{minUtil}.}
	\label{fig:Runtime2}	
\end{figure*}

As shown in Fig. \ref{fig:Runtime1} and Fig. \ref{fig:Runtime2}, the runtimes of all implemented algorithms decreased along with the increasing of the $minUtil$ or $minPro$ threshold. Specifically, the implemented algorithms with variants of pruning strategies greatly improved the performance, up to nearly one or two orders of magnitude faster than the baseline approach. For example, HUPNU$ _{P1234}$, which adopts all the efficient pruning strategies, outperformed the other variants of the designed approach. The reason is that the HUPNU$ _{P1234}$ algorithm is concerned with all the pruning strategies, and the unpromising candidates can be greatly reduced. Therefore, the traversal procedure to explore the unpromising patterns in the enumeration tree can be avoided, as well as the costly join operation to generate the huge unpromising candidates of the PU$^{\pm}$-lists. When the \textit{minUtil} and \textit{minPro} values were set lower, more and longer patterns were mined, and there was a greater computational cost in terms of the runtime that was required. This situation can be easily observed in a very dense dataset, for example in accidents and psumb.

The developed PU$^{\pm}$-list can also easily help the variants of the algorithms to directly mine the required patterns without candidate generation. The list structure can effectively reduce the multiple database scans. We can also observe that the designed pruning strategies can greatly help with reducing the number of unpromising patterns. The required memory usage of the variants of algorithms is much less, but due to the page limit, we omit some details of the results. However, in the general pattern-mining approach, it can be easily concluded that more memory usage was required when the threshold was set lower. The designed PU$^{\pm}$-list could actually solve this limitation by using a more compressed search space, and achieve effectiveness and efficiency for mining the PHUIs.

\subsection{Evaluation of Number of Patterns (Visited Nodes)}
In this section, the numbers of visited nodes (also known as the candidate patterns) in the designed PU$^{\pm}$-tree are compared to evaluate the effects  of pruning strategies. Note that the number of visited nodes of the four variants of the HUPNU algorithm is denoted as $N_{1}$, $N_{2}$, $N_{3}$, and $N_{4}$, respectively. According to the same parameter settings in Fig. \ref{fig:Runtime1} and Fig. \ref{fig:Runtime2}, the final results of the patterns are respectively shown in Table \ref{table:patterns1} and Table \ref{table:patterns2} in terms of \textit{minUtil} and \textit{minPro} thresholds. 

From Table \ref{table:patterns1} and Table \ref{table:patterns2}, it can be seen that $N_{1} > N_{2} > N_{3} > N_{4} >$ \textit{PHUIs} in terms of varied \textit{minUtil} and \textit{minPro}. We can then draw a conclusion such that: (1) The number of designed PHUIs discovered was fewer in the uncertain dataset compared to the number of candidate patterns. (2) The set of the discovered PHUIs was more meaningful; the designed HUPNU discover more concrete and useful patterns with both positive and negative unit utilities of item constraints than the traditional algorithms in HUPM. (3) The search space of the HUPNU was very large without any pruning strategies, but the developed strategies could reduce the its size. (4) The results were reasonable, and even the size of the search space could be reduced by different pruning strategies, the completeness, and the correctness of the final PHUIs could still be held. 

For the designed pruning strategies 2, 3, and 4 of the implemented HUPNU$ _{P1}$ algorithm, they hold the upper-bound values for the utility and probability of the patterns, and thus combinational exploration could be avoided. However, some unpromising patterns could not be effectively filtered, which can be obviously seen from the gap between the number of final PHUIs and $N_{1}$. It can also be seen that PU-Prune had a great performance in removing unpromising patterns, which could be observed in $N_{1}$ and $ N_{2} $. This strategy avoids the construction progress for the numerous unpromising patterns. We can also see that the EUCP strategy made a great effort to reduce the search space. The relationships of $N_{1} > N_{2} > N_{3} > N_{4}$ were correctly held. When the threshold value was set lower, for example, $minUtil$ or $minPro$, the gap between different implemented algorithms of the visited patterns could become huge, and the effectiveness of the designed pruning strategies is held.

\begin{table*}[!htbp]
	\setlength{\abovecaptionskip}{0pt}
	\setlength{\belowcaptionskip}{0pt}
	\scriptsize
	\fontsize{7.5pt}{9pt}\selectfont
	\centering
	\caption{Derived patterns under varied \textit{minUtil} and fixed \textit{minPro}}
	\label{table:patterns1}
	\begin{tabular}{c|c|llllll}
		\hline\hline
		\multirow{2}*{\textbf{Dataset}}&
		\multirow{2}*{\textbf{Pattern}}
		&\multicolumn{6}{c}{\textbf{Threshold of \textit{minUtil}}}\\
		\cline{3-8}
		&&$ minUtil_1 $ & $ minUtil_2 $ & $ minUtil_3 $ & $ minUtil_4 $ &  $ minUtil_5 $ &  $ minUtil_6 $ \\ \hline
		&$N_1$ &109,475,579	&79,746,008	&61,798,333	&50,499,186	&41,197,757	&33,144,921 \\
		& $N_2$ &14,399,980	&12,180,884	&10,571,745	&9,254,740	&8,123,221	&7,125,357 \\
		\textbf{(a) kosarak} & $N_3$ &12,322,701	&10,611,814	&9,369,910	&8,286,105	&7,345,445	&6,505,338  \\
		& $N_4$ &8,067,669	&6,811,140	&5,945,495	&5,217,922	&4,603,638	&4,081,767  \\
		& \textbf{PHUIs} & \textbf{45,613}	& \textbf{35,662}	& \textbf{28,077}	& \textbf{24,317}	& \textbf{21,420}	& \textbf{18,843}  \\		    
		\hline
		
		&$N_1$ &194,512	&163,169	&138,399	&120,560	&107,620	&96,807 \\
		& $N_2$ &147,052	&119,802	&98,237	&83,198	&72,963	&64,390  \\
		\textbf{(b) accidents} & $N_3$ &146,426	&119,140	&97,507	&82,366	&72,087	&63,486  \\
		& $N_4$ &144,957	&117,671	&96,038	&80,897	&70,618	&62,017  \\
		& \textbf{PHUIs} & \textbf{6,331}	& \textbf{5,462}	& \textbf{4,603}	& \textbf{3,940}	& \textbf{3,493}	& \textbf{3,120}  \\		    
		\hline
		
		&$N_1$ &2,522,059	&1,814,071	&1,365,950	&1,184,769	&1,059,297	&849,484 \\
		& $N_2$ &1,265,985	&801,074	&590,923	&490,411	&434,835	&342,120 \\
		\textbf{(c) pumsb} & $N_3$ &1,238,118	&795,362	&583,778	&484,847	&423,414	&337,667  \\
		& $N_4$ &1,236,179	&793,423	&581,839	&482,908	&421,475	&335,732 \\
		& \textbf{PHUIs} & \textbf{8,648}	& \textbf{4,123}	& \textbf{3,114}	& \textbf{1,773}	& \textbf{1,388}	& \textbf{1,238}  \\		    
		\hline
		
		&$N_1$ &96,139,509	&48,350,235	&35,111,699	&27,307,914	&21,761,430	&17,948,378 \\
		& $N_2$ &21,876,139	&17,735,518	&15,582,757	&13,731,190	&11,813,940	&10,462,617 \\
		\textbf{(d) retail} & $N_3$ &21,078,117	&17,429,475	&15,360,149	&13,568,770	&11,696,475	&10,374,338  \\
		& $N_4$ &20,968,980	&17,332,324	&15,270,502	&13,484,318	&11,615,439	&10,296,705   \\
		& \textbf{PHUIs} & \textbf{19,278}	& \textbf{13,584}	& \textbf{10,655}	& \textbf{8,747}	& \textbf{7,423}	& \textbf{6,378}  \\		    
		\hline	
		
		&$N_1$ &1,558,452	&974,145	&721,233	&507,260	&380,254	&322,867 \\
		& $N_2$ &828,898	&527,865	&394,426	&278,595	&215,407	&181,517 \\
		\textbf{(e) mushroom} & $N_3$ &700,121	&440,243	&327,442	&229,328	&175,038	&146,840 \\
		& $N_4$ &698,969	&439,522	&326,867	&228,871	&174,672	&146,516   \\
		& \textbf{PHUIs} & \textbf{147,088}	& \textbf{85,133}	& \textbf{62,773}	& \textbf{42,466}	& \textbf{30,510}	& \textbf{25,590} \\		    
		\hline
		
		&$N_1$ &14,059,513	&7,194,329	&4,743,728	&3,475,577	&2,739,103	&2,257,562 \\
		& $N_2$ &6,199,289	&2,365,166	&1,264,020	&820,482	&607,506	&491,439 \\
		\textbf{(f) T10I4D100K} & $N_3$ &4,619,226	&1,655,006	&872,646	&584,661	&453,250	&383,900  \\
		& $N_4$ &4,454,919	&1,570,317	&802,457	&519,180	&390,542	&323,262   \\
		& \textbf{PHUIs} & \textbf{37,242}	& \textbf{20,710}	& \textbf{13,916}	& \textbf{10,140}	& \textbf{7,854}	& \textbf{6,292}  \\		    
		\hline\hline
	\end{tabular}
\end{table*}

\begin{table*}[!htbp]
	\setlength{\abovecaptionskip}{0pt}
	\setlength{\belowcaptionskip}{0pt}
	\scriptsize
	\fontsize{7.5pt}{9pt}\selectfont
	\centering
	\caption{Derived patterns under varied \textit{minPro} and fixed \textit{minUtil}}
	\label{table:patterns2}
	\begin{tabular}{c|c|llllll}
		\hline\hline
		\multirow{2}*{\textbf{Dataset}}&
		\multirow{2}*{\textbf{Pattern}}
		&\multicolumn{6}{c}{\textbf{Threshold of \textit{minPro}}}\\
		\cline{3-8}
		&&$ minPro_1 $ & $ minPro_2 $ & $ minPro_3 $ & $ minPro_4 $ &  $ minPro_5 $ &  $ minPro_6 $ \\ \hline
		
		&$N_1$ &79,746,008	&73,064,534	&66,950,091	&61,745,111	&57,114,432	&53,014,125 \\
		& $N_2$ &12,180,884	&11,778,729	&11,445,291	&11,154,865	&10,808,309	&10,460,194 \\
		\textbf{(a) kosarak} & $N_3$ &10,611,814	&10,375,295	&10,205,956	&10,077,832	&9,882,828	&9,664,443  \\
		& $N_4$ &6,811,140	&6,028,964	&5,459,093	&5,019,367	&4,612,615	&4,243,295  \\
		& \textbf{PHUIs} & \textbf{35,662}	& \textbf{24,589}	& \textbf{17,304}	& \textbf{12,595}	& \textbf{9,362}	& \textbf{7,154}  \\		    
		\hline
		
		&$N_1$ &194,512	&150,179	&112,717	&84,525	&64,389	&48,824 \\
		& $N_2$ &147,052	&109,522	&81,229	&60,283	&44,655	&32,550  \\
		\textbf{(b) accidents} & $N_3$ &146,426	&109,127	&80,933	&60,098	&44,554	&32,490  \\
		& $N_4$ &144,957	&107,740	&79,657	&59,029	&43,610	&31,651  \\
		& \textbf{PHUIs} & \textbf{6,331}	& \textbf{4,552}	& \textbf{3,321}	& \textbf{2,411}	& \textbf{1,735}	& \textbf{1,237}  \\		    
		\hline
		
		&$N_1$ &2,522,059	&1,892,203	&1,433,997	&1,054,904	&780,328	&600,511 \\
		& $N_2$ &1,265,985	&904,469	&654,212	&473,337	&344,866	&256,833 \\
		\textbf{(c) pumsb} & $N_3$ &1,238,118	&885,941	&642,025	&466,156	&340,510	&254,184  \\
		& $N_4$ &1,236,179	&883,884	&640,104	&464,472	&338,917	&252,599  \\
		& \textbf{PHUIs} & \textbf{8,648}	& \textbf{5,582}	& \textbf{3,630}	& \textbf{2,277}	& \textbf{1,396}	& \textbf{817}  \\		    
		\hline
		
		&$N_1$ &37,793,976	&37,017,817	&36,106,601	&35,111,699	&34,139,328	&33,094,290 \\
		& $N_2$ &16,931,811	&16,454,728	&15,986,124	&15,582,757	&15,288,602	&15,031,344 \\
		\textbf{(d) retail} & $N_3$ &16,526,029	&16,131,265	&15,715,903	&15,360,149	&15,101,236	&14,876,214  \\
		& $N_4$ &16,507,082	&16,090,920	&15,652,668	&15,270,502	&14,984,007	&14,728,697   \\
		& \textbf{PHUIs} & \textbf{14,406}	& \textbf{13,048}	& \textbf{1,1802}	& \textbf{10,655}	& \textbf{9,660}	& \textbf{8,724}  \\		    
		\hline	
		
		&$N_1$ &974,145	&923,770	&860,295	&796,159	&727,630	&665,124 \\
		& $N_2$ &527,865	&485,700	&444,208	&404,089	&363,959	&328,273 \\
		\textbf{(e) mushroom} & $N_3$ &440,243	&406,109	&373,114	&341,243	&309,047	&281,344 \\
		& $N_4$ &439,522	&405,048	&372,035	&339,920	&307,708	&279,939   \\
		& \textbf{PHUIs} & \textbf{85,133}	& \textbf{75,769}	& \textbf{67,538}	& \textbf{60,065}	& \textbf{52,944}	& \textbf{47,398}  \\		    
		\hline
		
		&$N_1$ &7,194,329	&6,992,960	&6,804,076	&6,632,844	&6,457,696	&6,259,347 \\
		& $N_2$ &2,365,166	&2,195,551	&2,034,389	&1,879,105	&1,718,748	&1,549,820 \\
		\textbf{(f) T10I4D100K} & $N_3$ &1,655,006	&1,519,652	&1,398,078	&1,296,955	&1,211,962	&1,133,041  \\
		& $N_4$ &1,570,317	&1,384,445	&1,225,810	&1,097,950	&997,283	&905,622   \\
		& \textbf{PHUIs} & \textbf{20,710}	& \textbf{19,278}	& \textbf{17,869}	& \textbf{16,454}	& \textbf{15,132}	& \textbf{13,900}  \\		    
		\hline\hline
	\end{tabular}
\end{table*}

\subsection{Evaluation of Scalability}
From Fig. \ref{figOfScalability}, the scalability was carried out on a realistic BMS-POS dataset under varied dataset sizes. The threshold values were set as \textit{minPro} = 0.0001 and \textit{minUtil} = 10k, and the dataset size was varied from 100k to 500k. In Fig. \ref{figOfScalability}(a), we can then find that the runtimes of the four implemented algorithms linearly increased along with the increasing size of the dataset. We can see that the runtime of HUPNU$ _{P123}$ was close to that of HUPNU$ _{P12}$, and both of them were significantly faster than HUPNU$ _{P1}$. We also can see that HUPNU$ _{P1234}$ outperformed the other implemented algorithms. When the size of the dataset increases, the gap between the implemented algorithms becomes larger but they still remain stable with linear growth. The memory usage of  four implemented algorithms is shown in Fig. \ref{figOfScalability}(b). HUPNU$ _{P1}$ requires the most memory usage; HUPNU$ _{P123}$ and HUPNU$ _{P1234}$ has similar results, and requires the least memory usage. Fig. \ref{figOfScalability}(c) shows the number of visited patterns of the four implemented algorithms and the final PHUIs. The results show the effectiveness and efficiency of the designed pruning strategies, and when the size of the dataset becomes larger, the gap of those algorithms becomes huge.

\begin{figure*}[!htbp]
	\setlength{\abovecaptionskip}{0pt}
	\setlength{\belowcaptionskip}{0pt}
	\centering
	\includegraphics[trim=120 0 100 0,clip,scale=0.48]{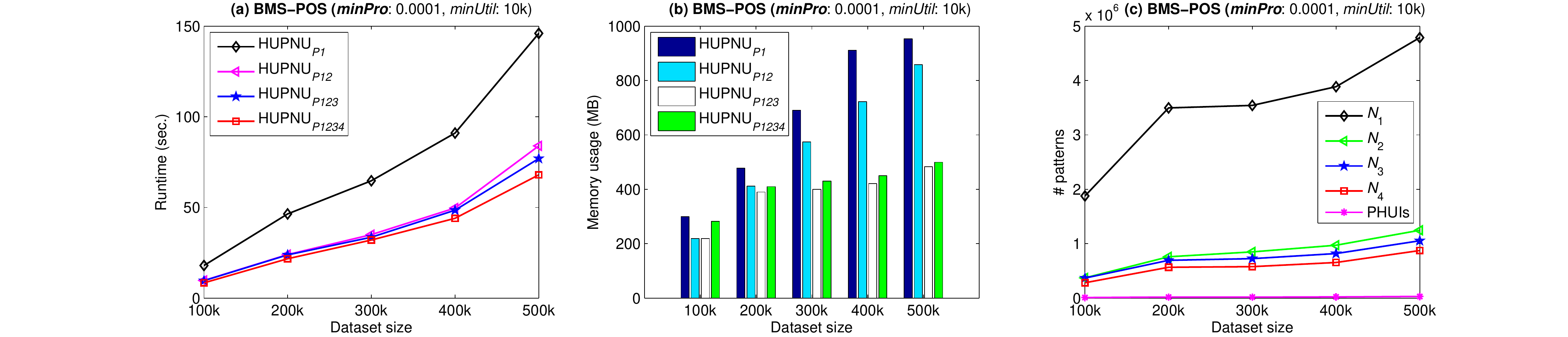}
	\caption{Scalability evaluation.}
	\label{figOfScalability}
	
\end{figure*}

\section{Conclusion} \label{sec:conclusion}

In this paper, we present a HUPNU algorithm by jointly considering the uncertainty and utility (both positive and negative) factors, to reveal the qualified high-utility patterns. This is the first work concerning these realistic factors in some real-life situations, such as Internet of Things data and manufacturing data. A vertical structure named PU$^{\pm}$-list was designed to keep necessary information, such as probability, and the positive and negative utilities of the items for later mining progress. Based on the above properties, the HUPNU algorithm could directly produce the qualified high-utility patterns in one phase. Moreover, several efficient pruning strategies were also developed to greatly reduce the search space for mining the promising patterns, and thus the computation could be sped up efficiently. Extensive experiments were carried out on several synthetic/realistic datasets to show the efficiency and effectiveness of the designed algorithm in terms of runtime, number of discovered qualified patterns, and scalability.

\ifCLASSOPTIONcaptionsoff
  \newpage
\fi



\bibliographystyle{IEEEtran}

\bibliography{IEEEabrv,paper}
%

%




\end{document}